\documentclass[11pt,reqno]{amsart}
\usepackage{amsmath,amssymb,graphicx,amscd}
\usepackage{mathrsfs}
\usepackage{enumerate}
\usepackage{color}

\usepackage{mathrsfs}
\usepackage{dsfont}
\usepackage{amsfonts}
\renewcommand{\P}{\mathbf{P}}

\newcommand{\E}{\mathbf{E}}
\newcommand{\I}{\mathbf{1}}
\newcommand{\FF}{\mathbb{F}}
\newcommand{\GG}{\mathbb{G}}
\newcommand{\F}{\mathcal{F}}
\newcommand{\G}{\mathcal{G}}

\newcommand{\dd}{{\mathrm{d}}}            

\newtheorem{thm}{Theorem}[section]
\newtheorem{cor}[thm]{Corollary}
\newtheorem{lem}[thm]{Lemma}
\newtheorem{prop}[thm]{Proposition}
\theoremstyle{definition}
\newtheorem{defn}[thm]{Definition}
\theoremstyle{remark}
\newtheorem*{rem}{Remark}

\theoremstyle{remark}

\numberwithin{equation}{section}

\numberwithin{equation}{section}

\begin{document}

\title{Hazard processes and martingale hazard processes}

\author{Delia Coculescu}
\address{ETHZ \\ Departement Mathematik, R\"{a}mistrasse 101\\ Z\"{u}rich 8092, Switzerland.}
\email{delia.coculescu@math.ethz.ch}

\author{Ashkan Nikeghbali}
\address{Institut f\"ur Mathematik,
Universit\"at Z\"urich, Winterthurerstrasse 190,
CH-8057 Z\"urich,
Switzerland}
\email{ashkan.nikeghbali@math.uzh.ch}


\begin{abstract}
In this paper, we provide a solution to two problems which have
been open in default time modeling in credit risk. We first show
that if $\tau$ is an arbitrary random (default) time such that its
Az\'ema's supermartingale $Z_t^\tau=\P(\tau>t|\F_t)$ is
continuous, then $\tau$ avoids stopping times. We then disprove a
conjecture about the equality between the hazard process and the
martingale hazard process,  which first appeared in
\cite{jenbrutk1}, and we show how it should be modified to become
a theorem. The pseudo-stopping times, introduced in
\cite{AshkanYor}, appear as the most general class of random times
for which these two processes are equal. We also show that these
two processes always differ when $\tau$ is an honest time.

\end{abstract}

\maketitle

\section{Introduction}

Random times which are not stopping times have recently played an
increasing role in the modeling of default times in the
hazard-rate approach of the credit risk. Following
\cite{jenbrutk1}, \cite{elliotjeanbyor}, \cite{bsjeanblanc}, a
hazard rate  model may be constructed in two steps. We begin with
a filtered probability space $(\Omega,\F,\FF=(\F_t),\P)$
satisfying the usual assumptions. The default time $\tau$ is
defined as a random time (i.e., a nonnegative $\F$-measurable
random variable) which is not an $\FF$- stopping time). Then, a
second filtration $\GG=(\G_t)$ plays an important role for
pricing. This is obtained by progressively enlarging the
filtration $\FF$ with the random time $\tau$: $\GG$ is the
smallest filtration satisfying the usual assumptions, containing
the original filtration $\FF$, and for which $\tau$ is a stopping
time, such as explained in \cite{jeulin}, \cite{jeulinyor}. The
filtration $\GG$ is usually considered as the relevant filtration
to consider in credit risk models: it represents the information
available on the market. The enlargement of filtration provides a
simple formula to compute the $\GG$-predictable compensator of
the process $\I_{\tau\leq t}$, which is a fundamental process
in the modeling of default times. Note that an alternative and
more direct hazard-rate approach, which historically appeared
first, consists in introducing one single global filtration $\GG$
from the start, where  the default time is a totally inaccessible stopping time with a given intensity.   Major papers using
the intensity based framework are \cite{JarrTur95},
\cite{JarrLaTur97}, \cite {Lando98}, \cite{MadUnal98},
\cite{DuffSingl}.

Both hazard-rate
approaches mentioned above, i.e., the direct approach or the one
based on two different sets of filtrations, model the occurring of
the default as a surprise for the market, that is, the default
time is a totally inaccessible stopping time in the global market
filtration $\GG$. The technique of enlargements of 
filtrations appears to be a useful tool, since it allows to
compute easily the price of a derivative, using the hazard
process. It allows as well an explicit
construction of the default compensator (in section 3, we shall
give simple "universal" formulae for the compensator of
pseudo-stopping times and honest times ). For instance, one can
take into account the link between the default-free and the
defaultable assets, or the incomplete information about the firm
fundamentals, and thus construct the compensator in an endogenous
manner (\cite{DuffLand}, \cite{Kusuok}, \cite{elliotjeanbyor},
\cite{GuoJarrZeng}, \cite{cocgemjea08}, \cite{freyschmidt}).

We now briefly justify the use of non stopping times for default
times (see  \cite{cocjeanik08} for a more detailed analysis, where
no-arbitrage conditions are also studied).

Defaultable claims are defined by their maturity date, say T, and
their promised stream of cash flows through time. Typically these
consist of a promised face value, to be paid at maturity and a
stream  to be paid during the lifetime of the contract. We may
suppose that the promised claim is an $\F_T$-measurable random
variable, denoted by $P$, since the intermediary payments may
be invested in the default-free money market account. In addition,
there is a random time $\tau$ at which the default occurs, and
when a recovery payment $R \neq P$ is made, in replacement of the
promised one. The defaultable payoff is of the form:
\begin{equation}\label{defpayoff}
X=P\I_{\tau> T}+R\I_{\tau\leq T}.
\end{equation}

When constructing a model for the pricing of defaultable claims
issued by a particular firm, say XYZ, one can proceed in two
steps. First, one needs to model the value of the promised claim
$P$ , as well as the recovery claim $R$ at intermediary times
$0\leq t\leq T$. For this, one can use the traditional
default-free evaluation technique. For instance, the promised
claim can be a fixed amount of  dollars or commodity. The question
of the recovery, even though more complicated, depends on the
value of the contract's collateral (for instance a physical
asset), which can be assumed to be default-free. In this case,
default-free techniques may be applied. Another possibility is to
estimate recovery rates from historical default data. Without
regard of the technique chosen, we denote by $\FF$ the information
available to the modeler after the first step, i.e, the estimation
of the promised and recovery assets, as well as the other
available market information. We exclude  information about the
assets issued by the firm $XYZ$, even if it is available, since
this should be the output of our evolution procedure, rather than
the input. For instance, we consider that the filtration $\FF$
does not contain information about the price of a defaultable bond
issued by the firm XYZ, even though this bond might be traded.
Usually, this construction leads to the situation where $\tau$ is
not an $\FF$-stopping time. For instance, in the classical Cox
framework, the default time is defined as:
\[
\tau:=\inf\{t|\Lambda_t>\Theta\},
\]
where $\Lambda$ is $\FF$ predictable and increasing, and $\Theta$ is an exponential random variable independent from $\FF$. This situation is also common in default models with incomplete information.

In a second step, we define the global filtration $\GG$ (i.e., the
one to use for pricing claims of the type (\ref{defpayoff})) in
such a way that $\tau$ becomes a stopping time. We are thus in the
progressive enlargements of filtrations setting.

When the random time is not a stopping time, several quantities
play an important role in the analysis of the model. The most
fundamental object attached to an  arbitrary random time $\tau$ is
certainly  the supermartingale $Z_t^\tau=\P(\tau>t|\F_t)$, chosen
to be c\`adl\`ag, called the Az\'ema's supermartingale associated
with $\tau$ (\cite{azema}). In the credit risk literature, very
often the random time $\tau$ is given with extra regularity
assumptions, such as continuity or monotonicity of $Z_t^\tau$.
However, these assumptions were  not translated into properties of
the random time $\tau$. We shall try to clarify the link between
the assumptions about the process $Z_t^\tau$ and the properties of
the default time $\tau$, since it is crucial for the modeler to
select the properties of the random time which appear to be the
most sensible.

Two more processes, closely related to the Az\'ema supermartingale
$Z^\tau$ and the $\GG$ predictable compensator of $\I_{\tau\leq
t}$, are often used in the evaluation of defaultable claims: the
hazard process and the martingale hazard process, which we now
define.
\begin{defn}\label{def::hazards}
\begin{enumerate}
\item Let $\tau$ be  a random time such that $Z_t^\tau>0$, for all
$t\geq0$ (in particular $\tau$ is not an  $\FF$-stopping
time). The nonnegative stochastic process
$\left(\Gamma_t\right)_{t\geq0}$ defined by:
$$\Gamma_t=-\ln Z_t^\tau,$$is called the \emph{hazard process}.
\item Let $D_t=\I_{\tau\leq t}$. An $\FF$-predictable
right-continuous increasing process $\Lambda$ is called an
$\FF$-\emph{martingale hazard process} of the random time $\tau$
if   the process $\widetilde{M}_t=D_t-\Lambda_{t\wedge\tau}$
is a $\GG$ martingale.
\end{enumerate}
\end{defn}

We see that the martingale hazard process is only defined  up to
time $\tau$ and that the stopped martingale hazard process
is the $\GG$-predictable compensator of the process $D$. This has
two implications. First, several martingale hazard processes might
exist for a default time, even if the predictable compensator is
unique. Secondly, this representation allows the martingale hazard
process to be $\FF$-adapted as stated in the definition even if,
obviously, the compensator is only $\GG$-adapted. In the next
section we will characterize the situation where the martingale
hazard process is unique.

Another important problem is to   know under which conditions
the hazard process and the martingale hazard processes coincide:
this was object of a conjecture made in \cite{jenbrutk1}:

\noindent \textit{Conjecture:} Suppose that the process $Z_t^\tau$
is decreasing. If $\Lambda$ is continuous, then $\Lambda=\Gamma$.

We shall show that the problem was not well posed and we shall see
how it should be phrased in order to have the equality between the
hazard process and the martingale hazard process under some
general conditions. More generally, the aim of this paper is to
show that the general theory of stochastic processes provides a
natural framework to pose and to study the modeling of default
times, and that it helps solve in a simple way some of the
problems raised there.

The paper is organized as follows:

\noindent In section 2, we recall some basic facts from the
general theory of stochastic processes that will be relevant for
this paper.

\noindent  In section 3, we show that if $Z_t^\tau$ is continuous,
then $\tau$ avoids stopping times. We also see under which
conditions the martingale hazard process and the hazard process
coincide: the pseudo-stopping times, introduced in
\cite{AshkanYor}, appear there as the most general class of random
times for which these two processes are equal. Moreover, we prove
that for honest times, which form another remarkable class of
random times, the hazard process and the martingale hazard process
always differ.
\newline

\textbf{Acknowledgments}. We wish to thank Monique Jeanblanc for very helpful conversations and comments that  improved the first drafts of this paper.

\section{Basic facts}

Throughout this paper, we assume we are given a filtered
probability space $\left(\Omega,\F,\FF,\P\right)$  satisfying the
usual assumptions.

\begin{defn}
A random time $\tau$ is a nonnegative random variable
$\tau:\left(\Omega,\mathcal{F}\right)\rightarrow[0,\infty]$.
\end{defn}

When dealing with arbitrary random times, one often works under the
following conditions:
\begin{itemize}
\item Assumption $\mathbf{(C)}$: all $\left( \mathcal{F}_{t}\right) $-martingales are \underline{c}ontinuous (e.g:
the Brownian filtration).

\item Assumption $\mathbf{(A)}$: the random time
$\tau $ \underline{a}voids every $\left( \mathcal{F}_{t}\right) $%
-stopping time $T$, i.e. $\mathbf{P}\left[ \rho =T\right] =0$.
\end{itemize}When we refer to assumptions $\mathbf{(CA)}$, this will
mean that both the conditions $\mathbf{(C)}$ and $\mathbf{(A)}$
hold.

We also recall
the definition of the Az\'ema's supermartingale as well as some
important processes related to it:
\begin{itemize}
\item the $\left( \mathcal{F}_{t}\right) $ supermartingale
\begin{equation}
Z_{t}^{\tau }=\mathbf{P}\left[ \tau >t\mid \mathcal{F}_{t}\right]
\label{surmart}
\end{equation}%
chosen to be c\`{a}dl\`{a}g, associated to $\tau $\ by Az\'{e}ma
(\cite{azema});

\item the $\left( \mathcal{F}_{t}\right) $ dual optional and predictable
projections of the process $1_{\left\{ \tau \leq t\right\} }$,
denoted respectively by $A_{t}^{\tau }$ and $a_{t}^{\tau }$;

\item the c\`{a}dl\`{a}g martingale
\begin{equation*}
\mu _{t}^{\tau }=\mathbf{E}\left[ A_{\infty }^{\tau }\mid \mathcal{F}_{t}%
\right] =A_{t}^{\tau }+Z_{t}^{\tau }.
\end{equation*}
\end{itemize}

We also consider the Doob-Meyer decomposition of (\ref{surmart}):%
\begin{equation*}
Z_{t}^{\tau }=m_{t}^{\tau }-a_{t}^{\tau }.
\end{equation*}%
We note that the supermartingale $\left(Z_{t}^{\tau}\right)$ is
the optional projection of $\mathbf{1}_{[0,\tau[}$.

Let us also define very rigourously the progressively enlarged filtration $\GG$.

We enlarge the initial filtration $\left(
\mathcal{F}_{t}\right) $\ with the process $\left( \tau \wedge
t\right) _{t\geq 0}$, so that the new enlarged filtration $\left(
\mathcal{G}_{t}\right) _{t\geq 0}$\  is the
smallest filtration (satisfying the usual assumptions) containing $\left( \mathcal{F}_{t}\right) $\ and making $%
\tau $\ a stopping time, that is
$$\mathcal{G}_{t}=\mathcal{K}_{t+},$$ where
$$\mathcal{K}_{t}=\mathcal{F}_{t}\bigvee\sigma\left(\tau\wedge
t\right).$$ A very common situation encountered in default times
modeling is the   $(H)$ hypothesis framework: every
$\FF$-local martingale is also a $\GG$-local martingale. For
instance, this property is always satisfied when the default time
is a Cox time.

However, it is possible to introduce more general random times. We
recall the definition of pseudo-stopping times which extend the
$(H)$ hypothesis framework and which will play an important role
in the study of hazard processes and martingale hazard processes.

\begin{defn}[\cite{AshkanYor}]\label{def:pst}
We say that $\tau $ is a $\left( \mathcal{F}_{t}\right) $
pseudo-stopping time if for every$\ \left( \mathcal{F}_{t}\right)
$-martingale $\left( M_{t}\right) $ in $\mathcal{H}^{1}$, we have%
\begin{equation}
\mathbf{E}M_{\tau }=\mathbf{E}M_{0}.  \label{pta}
\end{equation}
\end{defn}
\begin{rem}
It is equivalent to assume that (\ref{pta}) holds for bounded
martingales, since these are dense in $\mathcal{H}^{1}$. It can also
be proved that then (\ref{pta}) also holds for all uniformly
integrable martingales (see \cite{AshkanYor}).
\end{rem}
The following characterization of pseudo-stopping times will be often used in the sequel:
\begin{thm}[\cite{AshkanYor}]\label{mainthm2}
The following four properties are equivalent:
\begin{enumerate}
\item $\tau $\ is a $\left( \mathcal{F}_{t}\right) $\ pseudo-stopping time,
i.e (\ref{pta}) is satisfied;

\item $\mu _{t}^{\tau }\equiv 1$, $a.s$

\item $A_{\infty }^{\tau }\equiv 1$, $a.s$

\item every $\left( \mathcal{F}_{t}\right) $\ local martingale $\left(
M_{t}\right) $ satisfies
\begin{equation*}
\left( M_{t\wedge \tau }\right) _{t\geq 0}\text{ }is\text{ }a\text{ }local%
\text{ }\left( \mathcal{G}_{t}\right) \text{ }martingale.
\end{equation*}

If, furthermore, all $\left( \mathcal{F}_{t}\right) $ martingales
are continuous, then each of the preceding properties is equivalent
to

\item
\begin{equation*}
\left( Z_{t}^{\tau }\right) _{t\geq 0}\text{ is a decreasing }\left(
\mathcal{F}_{t}\right) \ \text{predictable process}
\end{equation*}
\end{enumerate}
\end{thm}
\begin{rem}

Of course, every stopping time is a pseudo-stopping time by the
the optional sampling theorem. But there are many examples or
families of pseudo-stopping which are not stopping times (see
\cite{AshkanYor}). Similarly, all random times which ensure that
the $(H)$ hypothesis holds are pseudo-stopping times. But there
are pseudo-stopping times for which the $(H)$ hypothesis does not
hold (in particular those which are $\F_\infty$-measurable; see
\cite{AshkanYor} for construction and further characterizations of
pseudo-stopping times).
\end{rem}

The following classical  lemma will be very helpful: it indicates
the properties of the above processes under the assumptions
$\mathbf{(A)}$ or $\mathbf{(C)}$ (for more details or references,
see \cite{delmaismey} or \cite{Ashkansurvey}).
\begin{lem}\label{lem:evitement}
Under condition $\mathbf{(A)}$,  $A_{t}^{\tau }=a_{t}^{\tau }$ is continuous.

Under condition $\mathbf{(C)}$, $A^{\tau}$ is predictable (recall
that  under $\mathbf{(C)}$ the predictable and optional sigma
fields are equal) and consequently $A^{\tau}=a^{\tau}$.

Under conditions $\mathbf{(CA)}$, $Z^{\tau}$ is continuous.
\end{lem}

We give a first application of theorem \ref{mainthm2} and lemma
\ref{lem:evitement} to illustrate how the general theory of
stochastic processes shed a new light on default time modeling. It
is very often assumed in the literature on default times that
$\tau$ is a random time whose associated Az\'ema supermartingale
is continuous and decreasing.
\begin{prop}\label{clar}

Let $\tau$ be a random time that avoids stopping times. Then
$(Z_t^\tau)$ is continuous and decreasing if and only if $\tau$ is
a pseudo-stopping time.
\end{prop}
\begin{proof}

If $\tau$ is a pseudo-stopping, then from theorem \ref{mainthm2},
$Z_t^\tau=1-A_t^\tau$. If $\tau$ avoids stopping times, then it
follows from lemma \ref{lem:evitement} that $A^\tau$ is continuous
and consequently $Z^\tau$ is continuous.

Conversely, if $Z^\tau$ is continuous, and if $\tau$ avoids
stopping times, then from the uniqueness of the Doob-Meyer
decomposition, $Z_t^\tau=1-a_t^\tau$. But since $\tau$ avoids
stopping times, we have $a_t^\tau=A_t^\tau$ from lemma
\ref{lem:evitement} and hence $Z_t^\tau=1-A_t^\tau$. Consequently,
from theorem \ref{mainthm2}, $\tau$ is a pseudo-stopping time.

\end{proof}
\begin{rem}

We shall see a slight reinforcement of this theorem in the next
section: indeed, we shall prove that if $Z^\tau$ is continuous,
then $\tau$ avoids stopping times.
\end{rem}
\section{Main theorems}

First, we clarify a situation concerning the hazard process.
Indeed, in the credit risk literature, the $\GG$ martingale
$L_t\equiv\I_{\tau>t}e^{\Gamma_t}$ plays an important role (see
\cite{jenbrutk1} or \cite{bsjeanblanc}). But from definition
\ref{def::hazards}, the hazard process is defined only when
$Z_t^\tau>0$ for all $t\geq0$. We wish to show that nevertheless,
the martingale $(L_t)$ is always well defined. For this, it is
enough to show that on the set $\{\tau>t\}$,   $\Gamma_t=-\log
Z_t^\tau $ is always well defined. This is the case thanks to the
following result from the general theory of stochastic processes:
\begin{prop}[\cite{jeulin}, \cite{delmaismey}, p.134]\label{lestbiendefinie}
Let $\tau$ be an arbitrary random time. The sets
$\left\{Z^{\tau}=0\right\}$ and $\left\{Z_{-}^{\tau}=0\right\}$ are
both disjoint from the stochastic interval $[0,\tau[$, and have the
same lower bound $T$, which is the smallest stopping time larger
than $\tau$.
\end{prop}

The next proposition gives general conditions under which $\Gamma$
is continuous, which is generally taken as an assumption in the
literature on default times: indeed, when computing prices or
hedging, one often has to integrate with respect to $\Gamma$ (see
\cite{jenbrutk1}, \cite{elliotjeanbyor} or \cite{bsjeanblanc}).
\begin{prop}
Let $\tau$ be a random time.
\begin{enumerate}[(i)]
\item Then under $(\textbf{CA})$, $(\Gamma_t)$ is continuous and
$\Gamma_0=0$. \item If $\tau$ is a pseudo-stopping time and if
$(\textbf{A})$ holds, then $(\Gamma_t)$ is a continuous increasing
process, with $\Gamma_0=0$.
\end{enumerate}
\end{prop}
\begin{proof}
This is a consequence of Lemma \ref{lem:evitement} and theorem \ref{mainthm2}.
\end{proof}

Now, what can one say about the random time $\tau$ if one assumes
that its associated Az\'ema's supermartingale is continuous? It
seems to have been an open question in the literature on credit
risk modeling for a few years now. The next proposition answers
this question:
\begin{prop}\label{prop::cont}

Let $\tau$ be a finite random time such that its  associated
Az\'ema's supermartingale $Z^\tau_t$ is continuous. Then $\tau$
avoids stopping times.
\end{prop}
\begin{proof}

It is known that $$Z^\tau_t=\; ^o(\I_{[0,\tau)}),$$that is
$Z^\tau_t$ is the optional projection of the stochastic interval
$[0,\tau)$. Now, following Jeulin-Yor \cite{jeulinyor}, define
$\widetilde{Z}_t$ as the optional projection of the stochastic
interval $[0,\tau]$: $$\widetilde{Z}_t=\;^o(\I_{[0,\tau]}).$$ It
can be shown (see \cite{jeulinyor}) that
$$\widetilde{Z}_+=Z^\tau\;\;\text{ and }\;\; \widetilde{Z}_-=Z^\tau_-.$$
Since $Z^\tau$ is continuous, we have
$$\widetilde{Z}_+=\widetilde{Z}_-=Z^\tau,$$ and consequently, for
any stopping time $T$:
$$\E[\I_{\tau\geq T}]-\E[\I_{\tau> T}]=0,$$
which means that $\P[\tau=T]=0$ for all stopping times $T$.
\end{proof}
As an application, we can  state the following enforcement of proposition \ref{clar}:
\begin{cor}
Let $\tau$ be a random time. Then $(Z_t^\tau)$ is a continuous and decreasing process if and only if $\tau$ is a pseudo-stopping time that avoids stopping times.
\end{cor}
Now we recall a theorem which is useful in constructing the
martingale hazard process. 
\begin{thm}[\cite{yorjeulin}]\label{calccomp}
Let $H$ be a bounded $\left(\mathcal{G}_{t}\right)$
predictable process. Then
$$H_{\tau}\mathbf{1}_{\tau\leq t}-\int_{0}^{t\wedge\tau}\dfrac{H_{s}}{Z_{s-}^{\tau}}da_{s}^{\tau}$$is
a $\left(\mathcal{G}_{t}\right)$ martingale.
\end{thm}

\begin{cor}
Let $\tau$ be a pseudo-stopping time
that avoids $\FF$ stopping times. Then
the
$\GG$ dual predictable projection of
$\mathbf{1}_{\tau\leq t}$ is
$\log\left(\frac{1}{Z_{t\wedge\tau}^{\tau}}\right)$.

Let $g$ be an honest time (that means that $g$ is the end of an
$\FF$ optional set) that avoids $\FF$ stopping times. Then the
$\GG$ dual predictable projection of $\mathbf{1}_{g\leq t}$ is
$A_t^g$.
\end{cor}

\begin{proof}
Let $\tau$ be a random time; taking $H\equiv1$, in Theorem \ref{calccomp} we find that
$\int_{0}^{t\wedge\tau}\frac{1}{Z_{s-}^{\tau}}dA_{s}^{\tau}$ is the
$\GG$ dual predictable projection of
$\mathbf{1}_{\tau\leq t}$.

When $\tau$ is a pseudo-stopping time
that avoids $\FF$ stopping times, we have
from Theorem \ref{mainthm2} that the
$\GG$ dual predictable projection of
$\mathbf{1}_{\tau\leq t}$ is
$-\log\left(Z_{t\wedge\tau}^{\tau}\right)$ since in this case $A_t^\tau=1-Z_t^\tau$ is continuous.

The second fact is an easy consequence of the well known fact that
the measure  $dA_t^g$ is carried by
$\left\{t:\;Z_{t}^{g}=1\right\}$ (see \cite{azema}).
\end{proof}

As a consequence, we have the following characterization of the
martingale hazard process: 
\begin{prop}\label{prop:marthazproc}Let $\tau$ be a random time. Suppose that $Z^\tau_t>0$, $\forall t$. Then, there exists a unique martingale hazard process $\Lambda_t$, given by:
$$\Lambda_t=\int_0^t \dfrac{da_u^\tau}{Z_{u-}},$$
where recall that $a_t^\tau$ is the dual predictable projection of $\I_{\tau\leq t}$.
\end{prop}
\begin{proof}
We suppose there exist two different martingale hazard processes $\Lambda^{1}$ and $\Lambda^{2}$ and denote
\[
T\left(  \omega\right)  =\inf\left\{  t:\Lambda_{t}^{1}\left(
\omega\right) \neq\Lambda_{t}^{2}\left(  \omega\right)  \right\}  .
\]
$T$ is an $\left(  \mathcal{F}_{t}\right)  $-stopping time hence a
$\GG$ stopping time. Due to the uniqueness of the predictable
compensator we must have for all $t\geq0:$
\[
\Lambda_{t\wedge\tau}^{1}=\Lambda_{t\wedge\tau}^{2}\ a.s.
\]
Hence, $T>\tau$ $a.s.$ and hence $Z_{t}^{\tau}=0$, $\forall t\geq T$. By assumption, this
is impossible, hence $\Lambda^{1}=\Lambda^{2}\ a.s.$
\end{proof}


It is conjectured in \cite{jenbrutk1} that if $\tau$ is any random
time (possibly a stopping time) such that $\P(\tau\leq t|\F_t)$ is
an increasing process, and if the martingale hazard process
$\Lambda$  is continuous, then $\Lambda=\Gamma$, where  $\Gamma$
is the hazard process. We now provide a counterexample to this
conjecture. Indeed, let $\tau$ be a totally inaccessible stopping
time of the filtration $\FF$. Then of course $\P(\tau\leq
t|\F_t)=\I_{\tau\leq t}$ is an increasing process. Let now $(A_t)$
be the predictable compensator of $\I_{\tau\leq t}$. It is well
known (see \cite{azema} or \cite{jeulin} for example) that $A_t$
is a continuous process (that satisfies $A_t=A_{t\wedge\tau}$) and
hence $\Lambda_t=A_t$ is continuous. But clearly
$\Gamma_t\neq\Lambda_t$.

We propose the following theorem instead of the above conjecture
(recall that the fact that Az\'ema's supermartingale  is
continuous and decreasing means that $\tau$ is a pseudo-stopping
time):
\begin{thm}
Let $\tau$ be a pseudo-stopping time. Assume further that $Z_t^\tau>0$ for all $t$.
\begin{enumerate}[(i)]
\item Under \textbf{(A)}, $\Gamma$ is continuous and $\Gamma_{t}=\Lambda_t=-\ln Z_{t}$.
\item Under \textbf{(C)}, if $\Lambda$ is continuous, then $\Gamma_{t}=\Lambda_t=-\ln Z_{t}$.
\end{enumerate}
\end{thm}
\begin{proof}
(i) follows from lemma \ref{lem:evitement}, Theorem \ref{mainthm2} and proposition \ref{prop:marthazproc}.

(ii) Assume \textbf{(C)} holds.  Since $\Lambda$ is assumed to be continuous, it follows from proposition \ref{prop:marthazproc} (2) that $a_t^{\tau}$ is continuous. Hence $\tau$ avoids all predictable stopping times. But under \textbf{(C)}, all stopping times are predictable. Consequently $\tau$ avoids all stopping times and we apply part (i).
\end{proof}
It has been proved in \cite{jenbrutk1} that in general, even under the assumptions \textbf{(CA)}, the hazard process and the martingale hazard process may differ. The example they used was $g\equiv\sup\{t\leq1:\;W_t=0\}$, where $W$ denotes as usual the standard Brownian Motion. This time is a typical example of an honest time (i.e. the end of an optional set). We shall now show that this result actually holds for any honest time $g$ and compute explicitly the difference in this case. We shall need for this the following characterisation of honest times given in \cite{ashyordoob}:
\begin{thm}[\cite{ashyordoob}]
\label{multiplicatcarac} Let $g$\ be an honest time. Then, under the
conditions \textbf{(CA)}, there exists a unique continuous and nonnegative
local martingale $\left( N_{t}\right) _{t\geq 0}$,
with $N_{0}=1$ and $\lim_{t\rightarrow \infty }N_{t}=0$, such that:%
\begin{equation*}
Z_{t}^g=\mathbf{P}\left( g>t\mid \mathcal{F}_{t}\right)
=\dfrac{N_{t}}{\Sigma_{t}},
\end{equation*}where $\Sigma_{t}=\sup_{s\leq t}N_{s}$. The honest time $g$ is also given by:
\begin{eqnarray}
g &=&\sup \left\{ t\geq 0:\quad N_{t}=\Sigma_{\infty }\right\}  \notag \\
&=&\sup \left\{ t\geq 0:\quad \Sigma_{t}-N_{t}=0\right\} .
\label{defdeg}
\end{eqnarray}
\end{thm}
\begin{prop}\label{re}
Let $g$ be an honest time. Under \textbf{(CA)}, assume that
$\P(g>t|\F_t)>0$. Then there exists a unique strictly positive and
continuous local martingale $N$, with $N_0=1$ and
$\lim_{t\to\infty}N_t=0$, such that:
   $$\Gamma_t=\ln \Sigma_t-\ln N_t \text{ whilst } \Lambda_t=\ln\Sigma_t,$$where $\Sigma_t=\sup_{s\leq t}N_s$. Consequently, $$\Lambda_t-\Gamma_t=\ln N_t,$$ and $\Gamma\neq\Lambda$.
\end{prop}
\begin{proof}

From theorem \ref{multiplicatcarac}, there exists a unique
strictly positive continuous local martingale $N$, such that
$N_0=1$ and $\lim_{t\rightarrow \infty }N_{t}=0$, such that:
\begin{equation*}
Z_{t}^g=\mathbf{P}\left( g>t\mid \mathcal{F}_{t}\right)
=\dfrac{N_{t}}{\Sigma_{t}}.
\end{equation*} Now an application of It\^{o}'s formula yields:
$$\mathbf{P}\left( g>t\mid \mathcal{F}_{t}\right)=1+\int_{0}^{t}\dfrac{\dd N_s}{\Sigma_s}-\int_{0}^{t}\dfrac{N_s}{\Sigma_s^2}\dd \Sigma_s.$$But on the support of $(\dd \Sigma_s)$, we have $\Sigma_t=N_t$ and
hence:

$$\mathbf{P}\left( g>t\mid \mathcal{F}_{t}\right)=1+\int_{0}^{t}\dfrac{\dd N_s}{\Sigma_s}-\ln \Sigma_t.$$
From the uniqueness of the Doob-Meyer decomposition, we
deduce that the dual predictable projection of $\I_{g\leq t}$ is
$\ln \Sigma_t$. Now,applying proposition \ref{prop:marthazproc},
we have:
$$\Lambda_t=\int_{0}^{t}\dfrac{\dd (\ln \Sigma_s)}{\mathbf{P}\left( g>s\mid \mathcal{F}_{s}\right)}=\int_{0}^{t}\dfrac{\Sigma_s}{\Sigma_s N_s}\dd  \Sigma_s=\ln \Sigma_t,$$where we have again used the fact that the support of $(\dd \Sigma_s)$, we have $\Sigma_t=N_t$. The result of the proposition now follows easily.
\end{proof}

We shall now outline a nontrivial consequence of Theorem
\ref{multiplicatcarac} here. In \cite{azemjeulknightyor}, the
authors are interested in giving explicit examples of dual
predictable projections of processes of the form $\mathbf{1}_{L\leq
t}$, where $L$ is an honest time. Indeed, these dual projections are
natural examples of increasing injective processes (see
\cite{azemjeulknightyor} for more details and references). With
Theorem \ref{multiplicatcarac}, we have a complete characterization
of such projections, which are also very important in credit risk modeling:
\begin{cor}
Assume the assumption \textbf{(C)} holds, and let
$\left(C_{t}\right)$ be an increasing process. Then $C$ is the dual
predictable projection of $\mathbf{1}_{g\leq t}$, for some honest
time $g$ that avoids stopping times, if and only if there exists a
continuous local martingale $N_{t}$, with $N_0=1$ and $\lim_{t\to\infty}N_t=0$,
such that
$$C_{t}=\ln \Sigma_{t}.$$
\end{cor}
\begin{proof}

This is a consequence of theorem \ref{multiplicatcarac} and the
fact, established in the proof of proposition \ref{re}, that the
dual predictable projection of $\I_{g\leq t}$ is $\ln \Sigma_t$.
\end{proof}

\renewcommand{\refname}{References}

\end{document}